\documentclass[11pt]{article}

\usepackage{amsmath}
\usepackage{amssymb}
\usepackage{amsthm}
\usepackage{graphicx}
\usepackage{natbib}
\usepackage{booktabs}
\usepackage{color,soul}
\usepackage[left=25.4mm,right=25.4mm,top=30mm,bottom=25.4mm,a4paper]{geometry}
\usepackage{subfigure}
\theoremstyle{definition}
\newtheorem{theorem}{Theorem}

\newtheorem{proposition}[theorem]{Proposition}
\newtheorem{remark}[theorem]{Remark}

\newcommand{\E}{\mathbb E}
\newcommand{\e}{\mathrm e}

\newcommand{\D}{\mathrm{d}}
\newcommand{\F}{\mathcal F}

\begin{document}
\title{Probabilistic and statistical properties of moment variations and their use in inference and estimation based on high frequency return data}
\author{Kyungsub Lee\footnote{klee@unist.ac.kr, School of Business Administration, UNIST, UNIST-gil 50, Ulsan, Republic of Korea}}
\maketitle

\begin{abstract}
We discuss the probabilistic properties of the variation based third and fourth moments of financial returns as estimators of the actual moments of the return distributions.
The moment variations are defined under non-parametric assumptions with quadratic variation method but for the computational tractability, we use a square root stochastic volatility model for the derivations of moment conditions for estimations.
Using the S\&P 500 index high frequency data, the realized versions of the moment variations is used for the estimation of a stochastic volatility model.
We propose a simple estimation method of a stochastic volatility model using the sample averages of the variations and ARMA estimation.
In addition, we compare the results with a generalized method of moments estimation based on the successive relation between realized moments and their lagged values.
\end{abstract}
\smallskip
\noindent \textbf{Keywords.} high order moment, quadratic variation, stochastic voaltility, generalized method of moments\\

\section{Introduction}

In this paper, we discuss the properties of third and fourth moment variations and their realized versions of financial asset returns.
The realized third and fourth moments are defined based on quadratic variation methods as extensions of the definition of the realized variance which is a high-frequency data based estimator for the variance of asset returns.
The realized third and fourth moments are unbiased estimators of the third and fourth moments of return, respectively, under the martingale assumption of the return process.

The third and fourth moments of asset returns are hard to measure precisely in a conventional measure based on sample average.
In spite of the difficulty, the high moments played important roles in asset pricing, portfolio and risk managements.
The literature includes \cite{KrausLitzenberger}, \cite{HarveySiddique1999}, \cite{HarveySiddique}, \cite{BakshiKapadiaMadan}, \cite{KimWhite}, \cite{Brooks}, \cite{Christoffersen}, \cite{LeonRubioSerna}, \cite{Neuberger2012}, \cite{Romo} and \cite{ChoeLee}.

The variation based definitions of the realized third and fourth moments of a return process used in this paper are introduced by \cite{ChoeLee}.
The third moment variation is defined as the quadratic covariation between the return and its squared processes over a fixed time period and
the fourth moment variation is defined as the quadratic variation of the squared return process.
The variation based realized moments are finite sum approximations of these (co)variations multiplied by 1.5.
Similar quadratic variation based methods are well known for the estimation of the variance of returns and there is a growing literature during the last two decades, only a few of them are quoted here, see \cite{ABDL}, \cite{Barndorff2002a},\cite{Barndorff-Nielsen2004}, \cite{Hansen} and \cite{MyklandZhang}.

\cite{ChoeLee} focused on the risk-neutral properties of the high moments by deriving the option implied expectation of the moments.
In contrast, in this paper, we focus on the properties of the realized high moments of the return process.
Empirical and simulation suggests that the variation based realized moments are relatively efficient measurements for the actual third and fourth moments of the return distribution under a martingale condition for the return process.
The traditional measures of the high moments based on the sample averages of the cube and fourth power of the return are heavily influenced by outliers \citep{KimWhite, Neuberger2012},
and the relative efficient measurements based on quadratic variation methods will have better performances in estimations and statistical inferences.
Indeed, we show a strong evidence of negative skew in daily return by the realized third moment, which can be rejected under a test based on the sample average of the cubed daily return.

In the next, we explain the estimation methods for a stochastic volatility model using the realized second and third moments.
\cite{Bollerslev2002} proposed an estimation method of a stochastic volatility diffusion model based on integrated volatility which is equivalent to the quadratic variation of the return under the model.
They showed that even though the volatility process in latent, by the observability of the integrated volatility,
we are able to estimate the mean-reversion, long-run variance and volatility of volatility parameters based on the generalized method of moments (GMM).

Recently, similar approaches for the estimations are studied in \cite{Bollerslev2011} and \cite{Garcia2011} for stochastic volatility models and \cite{Da2014} for the tick-structure of price movements.
We have an additional high-frequency quantity compared with the existing estimations of stochastic volatility models, the realized third moment, which is especially useful to estimate the leverage parameter in stochastic volatility models, which determines the skew, thanks to the closed form formula for the expectation of the third moment variation.
Combining these results, we employ the GMM estimation for a stochastic volatility model.
Additionally, we proposed a simple method of the stochastic volatility model estimation using sample average and autoregressive moving average model estimation method.
Despite of its simplicity, simulation studies show that the result values of the simple estimation is close to the original parameter settings compared with the complicated GMM.

One of the interesting properties about the realized moments is about moment variation swaps.
For example, the third moment variation swap, introduced by \cite{ChoeLee}, is an instrument
that allows investors to take bets on or hedge the skew risk of an underlying asset.
For the previous similar work on trading the skew risk by swap contracts, see \cite{Schoutens}, \cite{Neuberger2012}, \cite{Kozhan2013} and \cite{Zhao}.
\cite{ChoeLee} demonstrated that the return distribution of the portfolio hedged by the third moment swap has thinner tail that the return distribution of the underlying asset and hence one can reduce the fat tail and skew risk.

The remainder of the paper is organized as follows:
In Section~\ref{Sect:moment}, we discuss the basic properties of the moment variations under a martingale condition and compute the biases when the return is not martingale.
Section~\ref{Sect:GMM} presents the estimation methods and results of a stochastic volatility model with the realized second and third moment.
Section~\ref{Sect:conclusion} concludes the paper. 

\section{Moment variation}\label{Sect:moment}

We recall the definition of the third and fourth moment variations of stochastic processes, introduced by \cite{ChoeLee}, and discuss their basic properties.
Consider a semimartingale log-return process $R_t = \log S_t - \log S_0$ with the corresponding asset price process $S$.
The time horizon $[0,t]$ can be a day, a month or any time interval to be observed.
The third and fourth moment variation processes are defined by $[R,R^2]$ and $[R^2]$, respectively,
where the brackets are used to denote the quadratic (co)variation processes.
From now on, we discuss their realized values of moment variation from high-frequency data to measure the high order moments of the return distribution.

First, we show that the unbiasedness of the high moment variations under a martingale assumption as a measure of the actual moments of the return distribution.
Note that, 
\begin{align}
R^3_t &= R_t R^2_t = \int_0^t R^2_{u} \D R_u + \int_0^t R_{u} \D R^2_u + [R,R^2]_t \nonumber\\
&= \int_0^t R^2_{u} \D R_u + \int_0^t R_{u} (2R_{u} \D R_u + \D [R]_u) + [R,R^2]_t \nonumber\\
&= 3\int_0^t R^2_{u} \D R_u + \int_0^t R_{u} \D [R]_u + [R,R^2]_t. \label{Eq:1}
\end{align} 
In addition,
$$R^2_t = 2\int_0^t R_{u} \D R_{u} + [R]_t$$
and since the covariation between $R$ and $[R]$ is zero, the covariation process between $R$ and $R^2$ is represented by
$$ [R,R^2]_t = 2\int_0^t R_{u} \D [R]_{u}.$$
The above equation also can be derived by applying It\={o}'s formula to $R^3$,
$$ R^3_t = 3\int_0^t R^2_{u} \D R_u + 3\int_0^t R_{u} \D [R]_{u} $$
and compare it with Eq.~\eqref{Eq:1}.
Finally, substituting $\int_0^t R_{u} \D [R]_{u}$ with $[R,R^2]_t/2$ in Eq.~\eqref{Eq:1}, we have
\begin{equation}
R^3_t = 3\int_0^t R^2_{u} \D R_u + \frac{3}{2}[R,R^2]_t. \label{Eq:3}
\end{equation}

Similarly,
\begin{align*}
R^4_t &= R^2_t\cdot R^2_t = 2 \int_0^t R^2_{u} \D R^2_u + [R^2]_t \\
&= 2 \int_0^t R^2_{u} (2R_{u} \D R_u + \D [R]_u) + [R^2]_t \\
&= 4 \int_0^t R^3_{u} \D R_u + 2 \int_0^t R^2_{u} \D [R]_u + [R^2]_t.
\end{align*}
Since
$$ [R^2]_t = 4\int_0^t R^2_{u} \D [R]_{u},$$
we have
\begin{equation}
R^4_t = 4 \int_0^t R^3_{u} \D R_u + \frac{3}{2} [R^2]_t.  \label{Eq:4}
\end{equation}
By assuming further that $R$ is a martingale and the stochastic integrations
$$  \int_0^t R^2_{u} \D R_u, \quad  \int_0^t R^3_{u} \D R_u$$
are martingales, by Eqs.~\eqref{Eq:3}~and~\eqref{Eq:4}, we have
\begin{align*}
\E[R^3_t] &= \frac{3}{2}\E \left[ [R,R^2]_t\right],\\
\E[R^4_t] &= \frac{3}{2}\E \left[ [R^2]_t\right].
\end{align*}
Thus, the third and fourth moment variations multiplied by 1.5 can be used as unbiased estimators of the third and fourth moments of return $R_t$ under the martingale conditions.
Note that in \cite{ChoeLee}, the unbiasedness of the third and fourth moment variations in the stochastic volatility model of \cite{Heston1993} was shown but in the above, the unbiasedness property is extended to any martingale model.

We start with general martingale models, but for computationality, 
consider the Heston-type square root stochastic volatility model:
\begin{align}
\D R_t &= \sqrt{V_t} \D W^s_t \label{Eq:Heston1} \\
\D V_t &= \kappa(\theta-V_t)\D t + \gamma \sqrt{V_t} \D W^v_t, \quad \D[W^s, W^v]_t = \rho \D t.\label{Eq:Heston2}
\end{align}
with the Feller condition $2\kappa \theta > \gamma^2$.
It is well known that
\begin{align}
\E[V_t] &= V_0 \e^{-\kappa t} + \theta (1-\e^{-\kappa t})  \label{Eq:mean_V}
\end{align}
and with some calculations, we have
\begin{equation}
\E[ [R]_{t} ] = \E\left[ \int_0^t V_s \D s \right] = \frac{1-\e^{-\kappa t}}{\kappa} (V_0 -\theta) + \theta t \label{Eq:mean_rsm}
\end{equation}
and
\begin{align}
\E[V_t^2] &= \e^{-2\kappa t} V_0^2 + \frac{2\kappa\theta + \gamma^2}{\kappa}(\e^{-\kappa t} - \e^{-2\kappa t}) V_0
+ \frac{\theta(2\kappa\theta + \gamma^2)}{\kappa}\left( \frac{1}{2} - \e^{-\kappa t} + \frac{1}{2}\e^{-2\kappa t}\right). \label{Eq:mean_squaredV}
\end{align}
In addition, throughout this paper, the following computation results are used.

\begin{proposition}\label{Prop:variance}
Under the assumption of Eqs.~\eqref{Eq:Heston1}~and~\eqref{Eq:Heston2}, we have\\
(i)
\begin{align*}
\E\left[V_t\int_0^t V_u \D u\right] 
={}& (V_0-\theta)\left(\theta + \frac{\gamma^2}{\kappa}\right)t\e^{-\kappa t} + \left(\frac{(V_0-\theta)(V_0-2\theta)}{\kappa} - \frac{\gamma^2 V_0}{\kappa^2}\right)\e^{-\kappa t} \\
&+ \left\{ -\frac{(V_0-\theta)^2}{\kappa} + \frac{\gamma^2}{\kappa^2}\left(V_0 -\frac{1}{2}\theta \right) \right\}\e^{-2\kappa t} + \theta^2 t + \frac{(V_0-\theta)\theta}{\kappa} + \frac{\gamma^2\theta}{2\kappa^2}.
\end{align*}
(ii)
\begin{align}
\E\left[ [R]_t^2 \right] ={}& \frac{2(\theta-V_0)}{\kappa}\left(\theta + \frac{\gamma^2}{\kappa}\right)t\e^{-\kappa t} + 2\left\{ \frac{\gamma^2 \theta}{\kappa^3} - \frac{(V_0-\theta)^2}{\kappa^2} \right\}\e^{-\kappa t} \nonumber\\
&+ \left\{\frac{1}{\kappa^2}(V_0-\theta)^2 - \frac{\gamma^2}{\kappa^3}\left(V_0 - \frac{1}{2}\theta \right) \right\}\e^{-2\kappa t} + \theta^2 t^2 + \left\{ \frac{2(V_0-\theta)\theta}{\kappa} + \frac{\gamma^2\theta}{\kappa^2} \right\} t \nonumber\\
&+ \frac{1}{\kappa^2}(V_0 - \theta)^2 + \frac{\gamma^2}{\kappa^3}\left(V_0 - \frac{5}{2}\theta \right),\label{Eq:mean_sqauredrsm}\\
={}& \frac{1}{\kappa^2}(\e^{-\kappa t} - 1)^2 V_0^2 \nonumber\\
&+ \left\{ -\frac{2}{\kappa}\left(\theta + \frac{\gamma^2}{\kappa} \right)t\e^{-\kappa t} + \frac{4\theta}{\kappa^2}\e^{-\kappa t} - \left(\frac{2\theta}{\kappa^2} + \frac{\gamma^2}{\kappa^3} \right)\e^{-2\kappa t}+ \frac{2\theta}{\kappa}t - \frac{2\theta}{\kappa^2} + \frac{\gamma^2}{\kappa^3}\right\}V_0 \nonumber\\
&+ \frac{2\theta}{\kappa}\left(\theta + \frac{\gamma^2}{\kappa} \right)t\e^{-\kappa t} + 2\left(\frac{\gamma^2\theta}{\kappa^3} - \frac{\theta^2}{\kappa^2} \right)\e^{-\kappa t} + \left(\frac{\theta^2}{\kappa^2} + \frac{\gamma^2\theta}{2\kappa^3} \right)\e^{-2\kappa t}\nonumber\\
&+ \theta^2 t^2 + \left(\frac{\gamma^2 \theta}{\kappa^2} - \frac{2\theta^2}{\kappa} \right)t + \frac{\theta^2}{\kappa^2} - \frac{5\gamma^2\theta}{2\kappa^3}.\label{Eq:mean_sqauredrsm2}
\end{align}
\end{proposition}
\begin{proof}
Let
\begin{align}\label{Eq:ODE1}
z(t) = \E\left[ \left( \int_0^t V_s \D s \right)^2 \right] = 2\E \left[ \int_0^t V_s \left( \int_0^s V_u \D u  \right) \D s \right] = 2 \int_0^t p(s) \D s
\end{align}
where
\begin{align*}
p(s) &= \E\left[V_s\int_0^s V_u \D u\right] = \E\left[\left(V_0 + \kappa\int_0^s(\theta-V_u)\D u + \gamma\int_0^s\sqrt{V_u}\D W^v_u \right)\left(\int_0^s V_u \D u \right) \right]\\
&= (V_0 + \kappa\theta s) \E \left[ \int_0^s V_u \D u \right] - \kappa \E\left[ \left(\int_0^s V_u \D u \right)^2\right] + \gamma \E \left[ \int_0^s V_u \D u \int_0^s \sqrt{V_u} \D W^v_u \right]\\
&= (V_0 + \kappa\theta s) \int_0^s \E[ V_u ] \D u - \kappa z(s) + \gamma \int_0^s w(u) \D u
\end{align*}
where
\begin{align*}
w(u) &= \E\left[ V_u \int_0^u \sqrt{V_k} \D W^v_k \right] \\
&= -\kappa \E \left[ \int_0^u V_k \D k \int_0^u \sqrt{V_k} \D W_k^v \right]
+ \gamma \E \left[\left( \int_0^u \sqrt{V_k} \D W_k^v \right)^2 \right]\\
&= -\kappa \E \left[ \int_0^u V_k \left( \int_0^u \sqrt{V_k} \D W_k^v \right) \D k \right] + \gamma \int_0^u ((V_0 - \theta)\e^{-\kappa k} + \theta )\D k \\
&= -\kappa \int_0^u w(k) \D k + \gamma \int_0^u ((V_0 - \theta)\e^{-\kappa k} + \theta )\D k.
\end{align*}
By solving the above integration equation, we have
$$ w(u) = \gamma(V_0 - \theta)u\e^{-\kappa u} + \frac{\gamma \theta}{\kappa}(1-\e^{-\kappa u})$$
and by solving 
$$
z(t) = 2 \int_0^t \left\{ (V_0 + \kappa\theta s) \int_0^s \E[ V_u ] \D u - \kappa z(s) + \gamma \int_0^s w(u) \D u \right\}\D s,
$$
we have the formula for the expectation of the square of the integrated variance.
In addition, using
\begin{align*}
\E\left[V_t\int_0^t V_u \D u\right] ={}& \frac{1}{2}\frac{\D z(t)}{\D t},
\end{align*}
we have the desired result.
\end{proof}

\begin{proposition}\label{Prop:tmv}
Under the assumption of Eqs.~\eqref{Eq:Heston1}~and~\eqref{Eq:Heston2}, we have the following moment condition for the third moment variation.
\begin{align*}
\textrm{(i) }\E[R_u V_u] &= \frac{\gamma\rho}{\kappa} \{ (V_0 -\theta)\kappa u \e^{-\kappa u} + \theta( 1-\e^{-\kappa u} )\}\\
\textrm{(ii) }\E[[R^2,R]_t] &= 2\int_0^t \E[R_u V_u] \D u \\
&= \frac{2\gamma \rho}{\kappa} \left[(V_0-\theta)\left( \frac{1-\e^{-\kappa t}(\kappa t +1)}{\kappa} \right) + \theta \left(\frac{\e^{-\kappa t}-1}{\kappa} + t \right) \right].
\end{align*}
\end{proposition}
\begin{proof}
Note that
\begin{align*}
\E[R_u V_u] &= \E \left[ \int_0^u \sqrt{V_s}\D W^s_s \left(V_0 + \int_0^u \kappa(\theta-V_s)\D s + \int_0^u \gamma \sqrt{V_s} \D W^v_s  \right)\right]\\
&= \E \left[ \kappa\theta u \int_0^u \sqrt{V_s}\D W^s_s - \kappa\int_0^u V_s\D s\int_0^u \sqrt{V_s}\D W^s_s + \int_0^u  \sqrt{V_s}\D W^s_s \int_0^u \gamma \sqrt{V_s} \D W^v_s \right]\\
&= \E \left[- \kappa\int_0^u V_s\D s\int_0^u  \sqrt{V_s}\D W^s_s + \int_0^u \sqrt{V_s}\D W^s_s \int_0^u \gamma \sqrt{V_s} \D W^v_s \right].
\end{align*}
For the second term inside the expectation, by It\={o}'s isometry, we have
\begin{align}
x(t) :={}& \E\left[\int_0^t \sqrt{V_s}\D W^s_s \int_0^t \gamma \sqrt{V_s} \D W^v_s \right] = \gamma\rho \int_0^t \E[V_s] \D s \nonumber\\
={}& \gamma\rho \left\{ \frac{(V_0-\theta)}{\kappa} \left( 1-\e^{-\kappa t} \right) + \theta t \right\} \label{Eq:x}
\end{align}
and for the first term, using the integration by part,
\begin{align*}
y(t) :={}& \E \left[\int_0^t V_s\D s\int_0^t \sqrt{V_s}\D W^s_s \right] \\
={}& \E \left[ \int_0^t \left( \int_0^s V_u \D u \right) \sqrt{V_s} \D W^s_s \right] + \E \left[ \int_0^t \left( \int_0^s  \sqrt{V_u} \D W^s_u\right) V_s \D s \right]\\
={}& \E \left[ \int_0^t V_s \left( \int_0^s  \sqrt{V_u} \D W^s_u \right) \D s \right].
\end{align*}
Furthermore,
\begin{align*}
\E \left[ V_s \int_0^s  \sqrt{V_u} \D W^s_u \right] &= \E \left[ \left(V_0 + \int_0^s \D V_u \right) \int_0^s  \sqrt{V_u} \D W^s_u \right]\\
&= \E\left[\int_0^s \kappa (\theta - V_u) \D u \int_0^s \sqrt{V_u} \D W^s_u\right] + \E
\left[\int_0^s \gamma \sqrt{V_u} \D W^v_u \int_0^s \sqrt{V_u} \D W^s_u \right]\\
&= -\kappa \E\left[\int_0^s V_u \D u \int_0^s \sqrt{V_u}\D W^s_u \right] + \E
\left[\int_0^s \gamma \sqrt{V_u} \D W^v_u \int_0^s \sqrt{V_u} \D W^s_u \right]\\
&= -\kappa y(s) + x(s).
\end{align*}
Thus,
$$ y(t) =  \int_0^t \left\{ -\kappa y(s) + x(s) \right\}\D s$$
and solving the integration equation, we have
\begin{equation}
y(t) = \frac{\gamma\rho}{\kappa^2} \left\{(V_0-\theta)(1-\e^{-\kappa t}-\kappa t \e^{-\kappa t}) + \theta(-1 + \e^{-\kappa t} + \kappa t)\right\}.
\end{equation}
Applying
\begin{align}
\E[R_u V_u] &= -\kappa y(u) + x(u) \label{Eq:RV},
\end{align}
we complete the proof.
\end{proof}

The formulas for $\E[[R]_t^2]$ and $\E[[R, R^2]_t]$ seems complicated but if we assume that the variance process is in the stationarity state at time 0, i.e., $V_0 = \theta$, then the formulas are simpler.
We discuss their uses in the next section.
The moment condition for the fourth moment variation, $\E[[R^2]_t]$ is rather complicated and not used elsewhere in this paper, so we discuss the derivation in Appendix~\ref{Sect:fmv}.

\begin{remark}
Consider the biases of the estimators when the drift in the return process is positive constant under the Heston-type square root stochastic volatility model:
\begin{align*}
\D R_t &= \mu \D t + \sqrt{V_t}\D W_t^s,\\
\D V_t &= \kappa(\theta - V_t) \D t + \gamma \sqrt{V_t} \D W_t^v, \quad \D [W^s, W^v]_t = \rho \D t.
\end{align*}
Note that the bias between the third moment and the expectation of $1.5[R,R^2]$ is
$$\E [R_t^3] - \frac{3}{2}\E[[R,R^2]_t] = 3 \E\left[\int_0^t R_u^2 \D R_u \right].$$
In addition,
\begin{align}
\E[R_u^2] &= \E \left[2\int_0^u R_s \D R_s + [R]_u \right]
= \E \left [2\int_0^u \mu R_s \D s + \int_0^u V_s \D s \right]
= \mu^2 u^2 + \theta u \label{Eq:R_square}
\end{align}
where we use the fact that the long-run expectation of $V_s$ is $\theta$ and assume that the variance process starts from the negative infinity and is in its stationary regime.
Thus, we have
\begin{align*}
3 \E \left[ \int_0^t R^2_u \D R_u \right] = 3 \E \left[ \int_0^t \mu R^2_u  \D u \right] = \mu^3 t^3 + \frac{3}{2}\mu\theta t^2
\end{align*}
which is the bias between the third moment and the variation based third moment.

Similarly, for the bias of the fourth moment, we have
\begin{align*}
\E [R^4_t] - \frac{3}{2}\E [[R^2_t]] = \E \left[ 4\int_0^t R^3_u \D R_u\right]
\end{align*}
and
\begin{align*}
\E[R_u^3] &= 3\E\left[\int_0^u R_s^2 \D R_s \right] + 3\E\left[\int_0^u R_s \D [R]_s \right]
= 3 \E \left[ \int_0^u \mu R_s^2 \D s \right] + 3 \E \left[\int_0^u R_s V_s \D s \right]\\
&= \mu^3 u^3 + \frac{3}{2}\mu\theta u^2 + \frac{3\gamma\rho\theta}{\kappa}\left(\frac{\e^{-\kappa u}-1}{\kappa} + u \right)
\end{align*}
where we use the formula for the long run expectation of $R_s V_s$ derived in Proposition~\ref{Prop:tmv}.
Thus, the bias with a constant drift in the return process is
$$ \E [R^4_t] - \frac{3}{2}\E [[R^2_t]] = \mu^4 t^4 + 2 \mu^2 \theta t^3 + \frac{12\gamma\rho\theta\mu}{\kappa}\left(\frac{1-\e^{-\kappa t}}{\kappa^2} + \frac{t^2}{2} - \frac{t}{\kappa} \right).$$

In Figure~\ref{Fig:bias}, we plot the term structures of the expected third moments with zero drift (solid) and constant drift (dashed) in the return process.
The parameter settings are $\theta = 0.0233, \kappa = 18.05, \gamma = 1.2305, \rho = -0.6191$ and $\mu = 0.05$.
These settings come from Section~\ref{Sect:estimation} where we perform an estimation for the S\&P 500 index series ranged from 2001 to 2007 (except for $\mu$, which is just presumed).
In addition, the theoretical formula for the long run expectation of the third moment will be derived in Eq.~\eqref{Eq:simple_rtm}.
The figure shows that the bias increases as the time horizon increases but it is rather insignificant especially for short term horizons.
\end{remark}

Second, the sample means of the moment variations multiplied by 1.5 are consistent estimators of the actual moments when the return process is a time homogeneous martingale with independent increment.
In this context, the sample mean implies the average of $M$ numbers of the moment variations with subintervals $[t_j, t_{j+1}]$ where $\Delta t =t_{j+1} - t_j$ is fixed for $j=0,\cdots M-1$.
For the third moment variation, we write
\begin{align*}
\frac{3}{2}[R,R^2]_{t_j, t_{j+1}} = R^3_{t_j, t_{j+1}} - 3\int_{t_j}^{t_{j+1}} R^2_{t_j, u} \D R_{t_j,u}
\end{align*}
where the time subscripts imply that the starting point of each process is $t_j$ and the end point is $t_{j+1}$.
For example, $R^2_{t_j, u} = \log^2 (R_{u} /R_{t_j}) $.

By taking the sample mean to the above equation, the first term of the r.h.s. is the sample mean of the third moment and hence converges to the actual third moment of the return distribution.
The sample mean of the second term of the r.h.s. represented by a stochastic integration converges to zero since the return process, the integrator, is a time homogeneous martingale with independent increments.
(We also use the additional assumption that the stochastic integration is also martingale.)
Thus, the sample means of the moment variations multiplied by 1.5 converges to the actual third moment as the sample size goes to infinity.
For the fourth moment variation, the same argument is applied to show the consistency.

Third, with the same assumptions in the above paragraph and by the central limit theorem, the sample means of the moment variations have asymptotically normal distributions as $[R,R^2]_{t_j, t_{j+1}}$ are i.i.d. random variable.
The derivation of variance of the asymptotic distribution is complicated and model dependent and 
we discuss the derivation of the variance of the third moment variation in Appendix~\ref{Sect:vtmv}.

Fourth, empirical studies in the next section and demonstrate that the moment variations are relatively efficient estimators compared with the conventional measures of the moments,
i.e., the variances of the moment variations are less than the variances of $R^3_t$ and $R^4_t$.
We also perform a simulation study to demonstrate the convergence of the third moment variation (multiplied by 1.5) to the theoretic value of the actual third moment.
Figure~\ref{Fig:Heston_rtm} shows that the sample mean of $\frac{3}{2}[R^2,R]_t$ converges to the theoretical expectation of the third moment faster than the sample mean of $R^3_t$ under the stochastic volatility model.
In the simulation, the parameter and initial value settings are $\kappa = 3, \theta = 0.04, \gamma = 2, \rho = -0.5$ and $V_0 = 0.05$.

\begin{figure}
\centering
  \includegraphics[width=10cm]{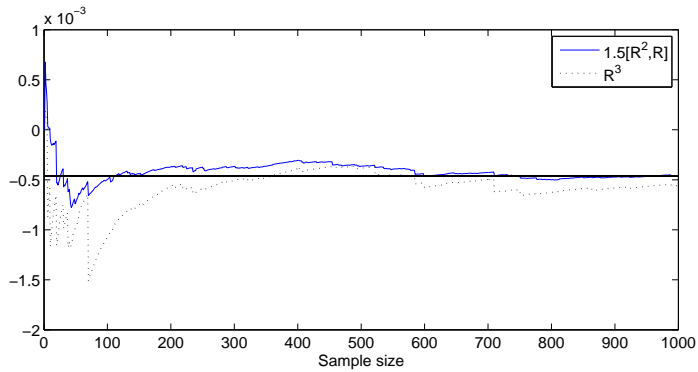}\\
  \caption{Moment estimator based on $[R^2,R]$ (dashed curve) converges to theoretical third moment (dashed straigth line) faster than the sample third moment (dotted) as sample size increase in the Heston-type stochastic volatility model}\label{Fig:Heston_rtm}
\end{figure}

\section{Realized moment variation and empirical study}\label{Sect:GMM}

\subsection{Definition and basic property}
Since the moment variations are not directly observable, we use the realized version of the variations to perform empirical study.
The realized third and fourth moments of $R_t$ are the realizations of the variations over a partition $0 = t_0< \cdots < t_N =t$ multiplied by $\frac{3}{2}$, i.e., the finite sum approximations represented by
\begin{align}
\frac{3}{2}\widehat {[R, R^2]}_t &= \frac{3}{2}\sum_{i=1}^{N} (R_{i} - R_{i-1})(R^2_{i} - R^2_{i-1}),\label{Eq:def_rm3}\\
\frac{3}{2}\widehat {[R^2]}_t &= \frac{3}{2}\sum_{i=1}^{N} (R^2_{i} - R^2_{i-1})^2,\label{Eq:def_rm4}
\end{align}
respectively. For simplicity, $R_i = R_{t_i}$ in the above equations.
By the semimartingale theory, we have
\begin{eqnarray*}
&\widehat {[R, R^2]}_t \rightarrow  [R,R^2]_t, \quad &\frac{3}{2}\E\left[ \widehat {[R, R^2]}_t \right] = \frac{3}{2}\E[[R,R^2]_t]= \E[R_t^3],\\
&\widehat {[R^2]}_t\rightarrow  [R^2]_t, \quad &\frac{3}{2}\E\left[\widehat {[R^2]}_t\right] = \frac{3}{2}\E[[R^2]_t] = \E[R_t^4]
\end{eqnarray*}
with limits in probability as the partition size goes to zero.
Similarly with the relationship between realized variance and integrated variance,
the realized moment variations are unbiased and consistent estimators of $1.5[R,R^2]_t$ and $1.5[R^2]_t$.
In addition, with the martingale condition for the return process, the realized third and fourth moment variations are unbiased estimators of the third and fourth moments of the return $R_t$, respectively.

\begin{figure}
\centering
\includegraphics[width=6cm]{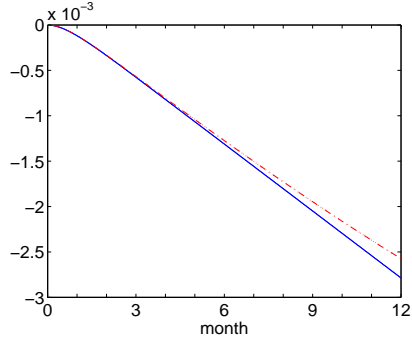}
\caption{Term structure of bias : expected third moment with zero drift (solid) and constant drift (dashed) under stochastic volatility model}
\label{Fig:bias}
\end{figure} 

Empirical study suggests that the realized moment variations are relatively efficient estimator compared with the conventional measures of the moments.
Using the S\&P 500 index series ranged from 2001 to 2007,
we plot the dynamics of the daily realized moments defined by Eqs.~\eqref{Eq:def_rm3}~and~\eqref{Eq:def_rm4} in Figures~\ref{Fig:rtm}~and~\ref{Fig:rfm}.
The finite sums are computed by five-minutes interval over one day horizon.
More precisely, the daily realized third moment variation is $1.5 \widehat {[R,R^2]}_{t_i,t_{i+1}}$ where the subscript implies that the quantity is observed from time $t_i$ to $t_{i+1}$ with each $t_{i+1} - t_i = 1$ day
and similarly for the daily realized fourth moment variation $1.5 \widehat {[R^2]}_{t_i,t_{i+1}}$.
In addition, the dynamics of realized $R^3_t$ and $R^4_t$ with $t=1$ day are plotted in Figures~\ref{Fig:r3}~and~\ref{Fig:r4}.

Let us compare Figures~\ref{Fig:rtm}~and~\ref{Fig:r3}.
The realized third moment variation varies from around $-6\times10^{-5}$ to $6\times10^{-5}$.
In contrast, the sample $R^3_t$ distributed from around $-1\times10^{-4}$ to $2\times10^{-5}$, implying much larger standard deviation in $R^3_t$ than the realized third moment variation.
Similarly, the sample $R^4_t$ in Figure~\ref{Fig:r4} has much larger standard deviation than the realized fourth moment variation in Figure~\ref{Fig:rfm}.
This observations imply that the realized moment variations $1.5 \widehat {[R,R^2]}$ and $1.5 \widehat {[R^2]}$ are relatively efficient estimators for the corresponding moments than the sample means of $R^3$ and $R^4$, respectively.

The summary statistics of $R^3_t$, $R^4_t$ and the realized moment variations are presented in Table~\ref{Table:statistics}.
The table shows that the sample standard deviations of the realized moment variations are about half of the standard deviations of $R^3_t$ and $R^4_t$.

\begin{figure}
\centering
  \includegraphics[width=10cm]{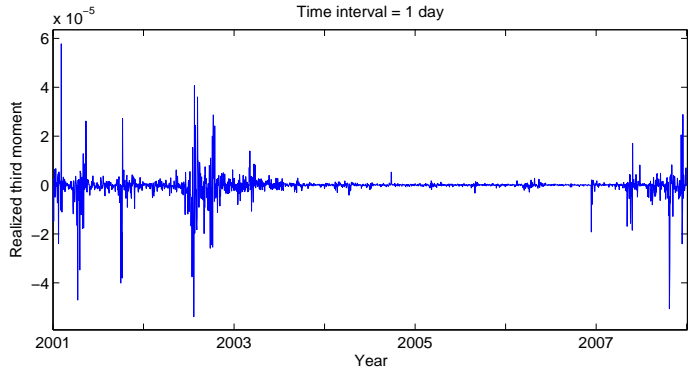}\\
  \caption{The dynamics of daily realized third moment variations of the S\&P 500 index from 2001 to 2007}\label{Fig:rtm}
\end{figure}

\begin{figure}
\centering
  \includegraphics[width=10cm]{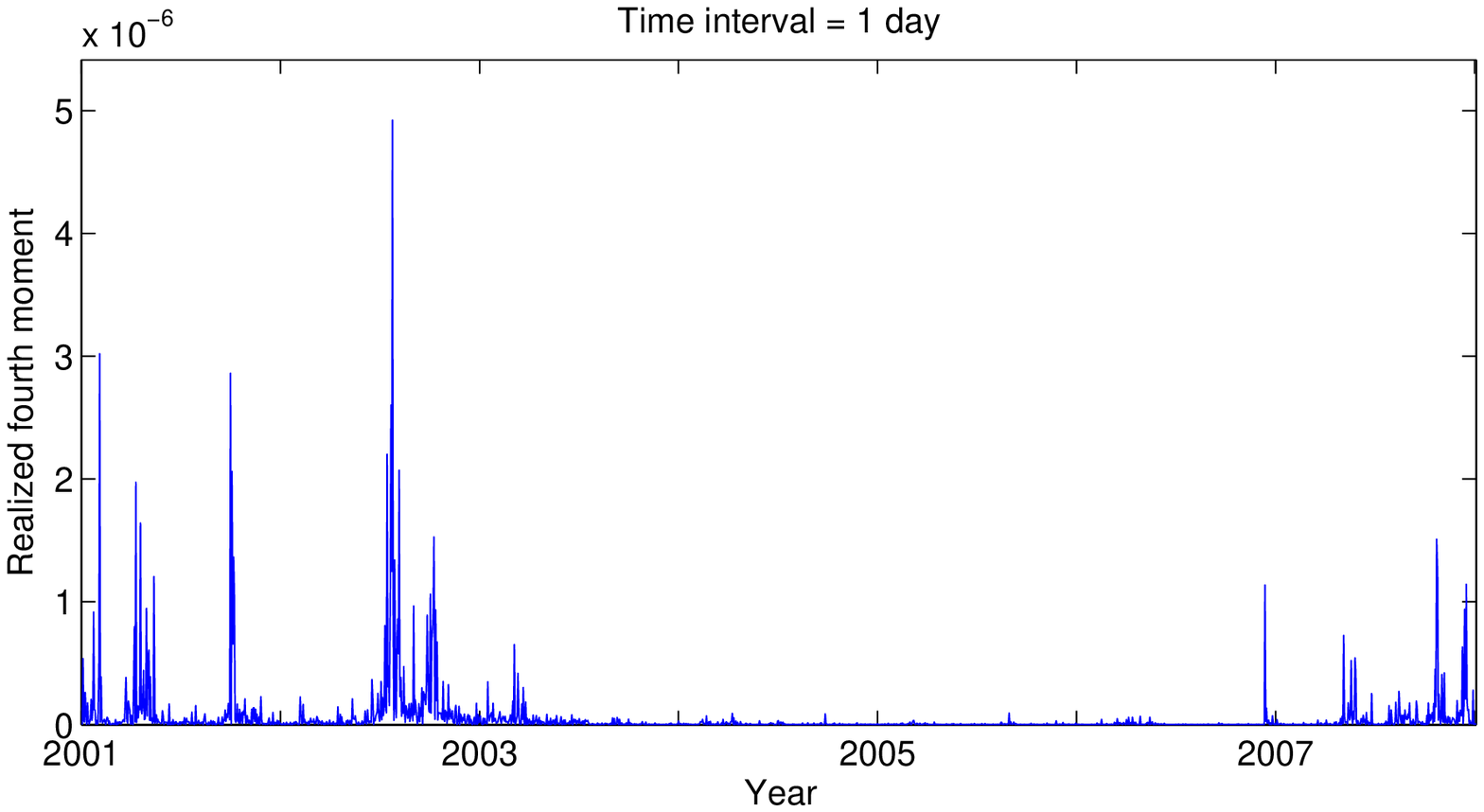}\\
  \caption{The dynamics of daily realized fourth moment variations of the S\&P 500 index from 2001 to 2007}\label{Fig:rfm}
\end{figure}

\begin{figure}
\centering
  \includegraphics[width=10cm]{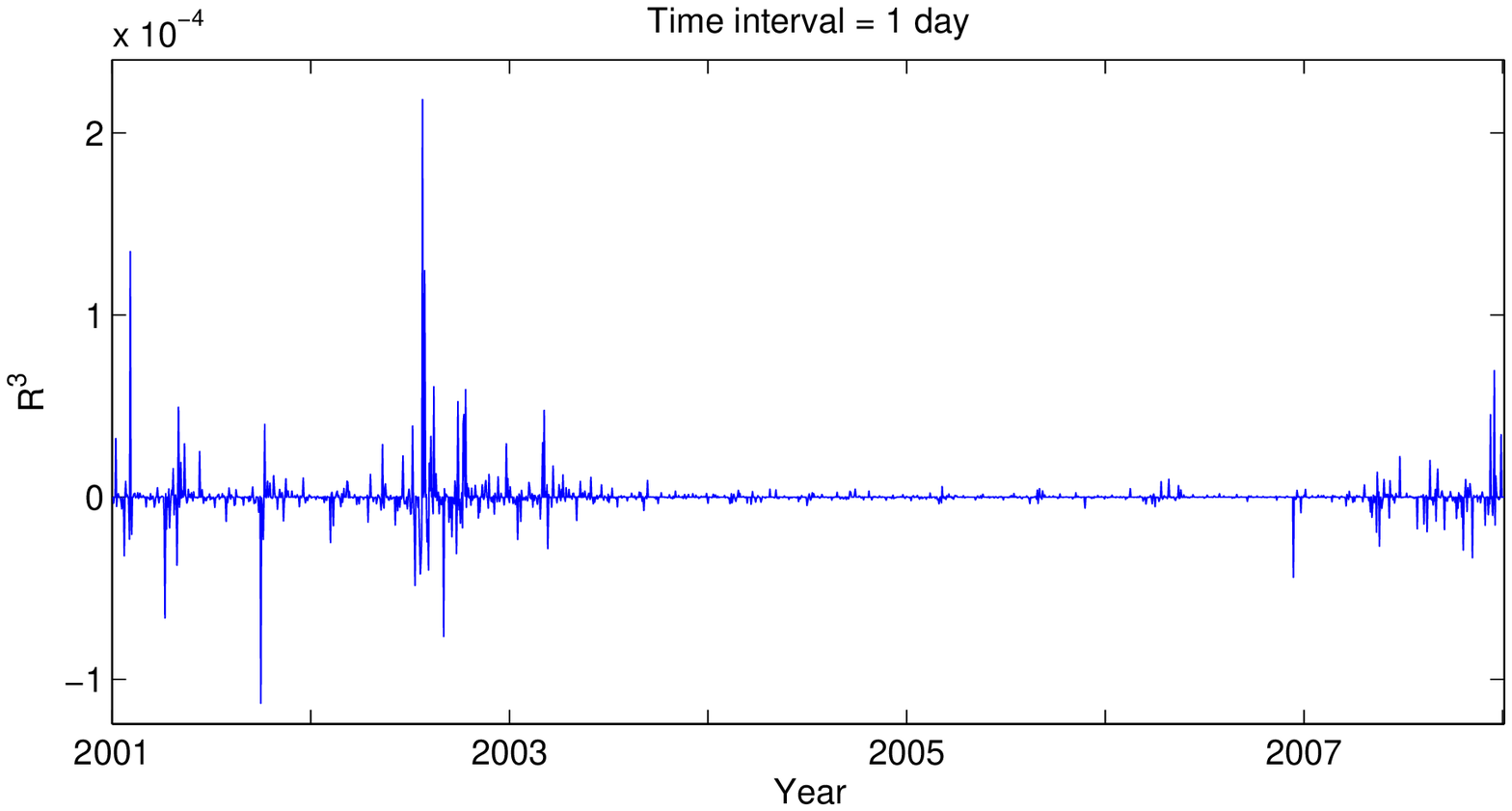}\\
  \caption{The dynamics of realized $R_t^3$ of the S\&P 500 index from 2001 to 2007}\label{Fig:r3}
\end{figure}

\begin{figure}
\centering
  \includegraphics[width=10cm]{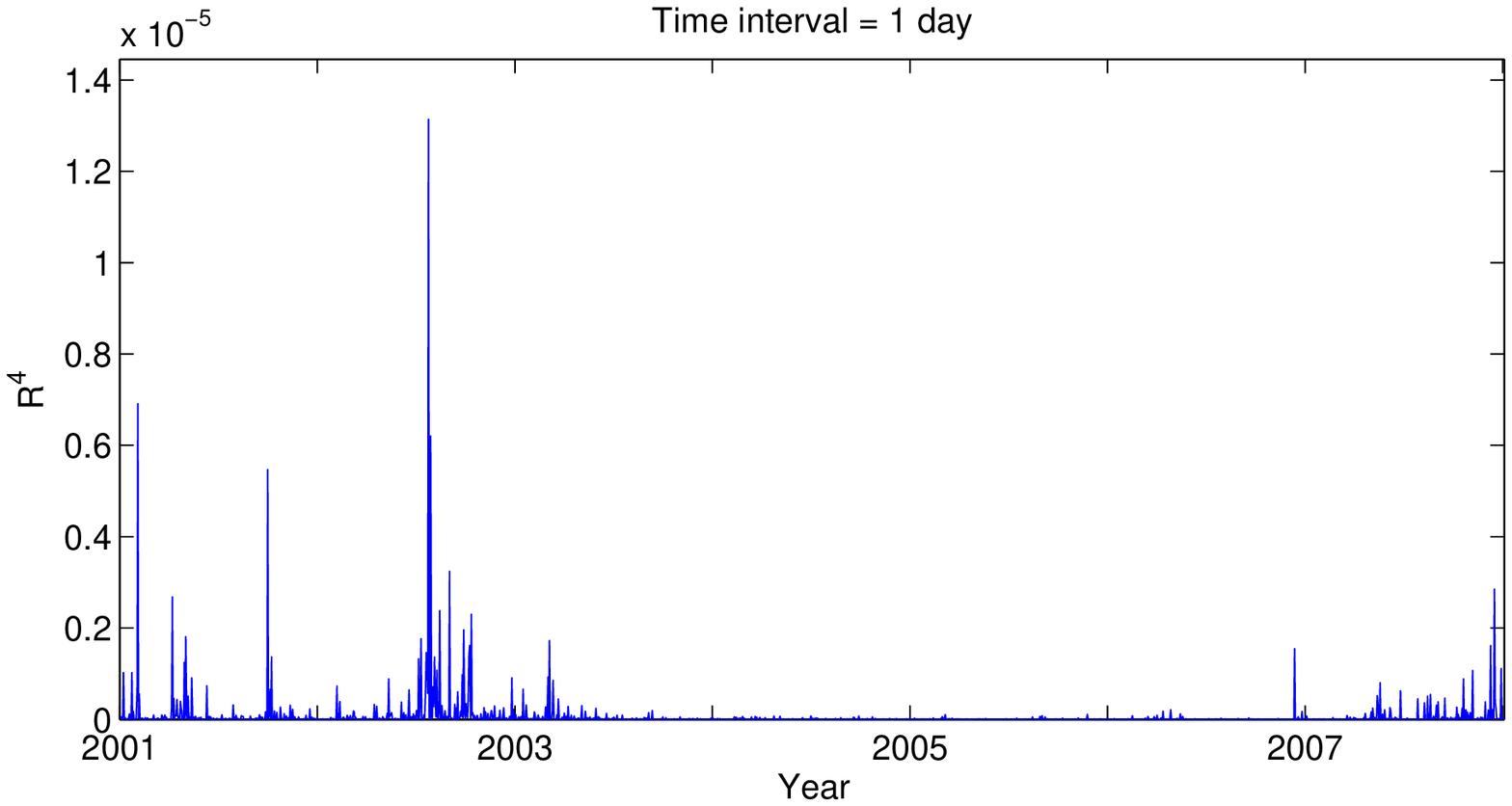}\\
  \caption{The dynamics of realized $R_t^4$ of the S\&P 500 index from 2001 to 2007}\label{Fig:r4}
\end{figure}

\begin{table}
\caption{The sample mean and standard deviations of realized moment variations, $R_t^3$ and $R_t^4$ based on the S\&P 500 index  from 2001 to 2007}\label{Table:statistics}
\centering
\begin{tabular}{ccccc}
  \hline
            & $[R^2,R]_t$ & $R_t^3$ & $[R^2]_t$ & $R_t^4$ \\
  \hline
  mean      & $-4.1626\times10^{-7}$ & $1.1697\times10^{-7}$ & $5.7238\times10^{-8}$ & $6.8150\times10^{-8}$ \\
  std. dev. & $ 4.9957\times10^{-6}$ & $9.7989\times10^{-6}$ & $2.3661\times10^{-7}$ & $4.5106\times10^{-7}$ \\
  \hline
\end{tabular}
\end{table}

The sample mean of $R^3_t$ is positive, meanwhile the sample mean of the realized third moment is negative.
The positive sample mean of $R^3_t$ contrasts with the fact that the financial asset return distribution is generally negatively skewed.
It is well known that the evidence for negative skewness of daily returns is rather weak compared with the returns over longer time horizons.
However, the sample mean of the realized third moment variation is still negative and this might support the negative skew of daily return.

We perform hypothesis tests to examine the negative skewness of the daily return distribution of the S\&P 500 index.
Obviously, the result of the $t$-test of the null hypothesis that $\E[R^3_t] =0$ against the alternative hypothesis that $\E[R^3_t] <0$ may not reject the null hypothesis with the exact level of significance $p\textrm{-value}= 0.6958$.
In contrast, the $p$-value of the $t$-test of the null hypothesis that $\E[\widehat{[R,R^2]}_t] =0$ against $\E[\widehat{[R,R^2]}_t] <0$ is $1.7823\times10^{-4}$ and we reject the null hypothesis at the level of significance 0.01 or even at the higher level.

Since the distribution of the realized third moments is not normal and there might be an error in the $t$-test.
We perform the \cite{Wilcoxon} signed rank test for the hypothesis that the realized third moment has zero median against the alternative that the distribution does not have zero median.
The null hypothesis is also rejected with the reported $p\textrm{-value}=0.0026$ and this is an evidence supporting negative skewness of the daily return.
This is an example that the realized third moment variation based inference have a better result than the case of using the cube of daily return.

\subsection{GMM estimation of SV model}\label{Sect:estimation}
In this section, we show estimation methods of a stochastic volatility model based on the realized second and third moment variations.
Our methodology is an extension of \cite{Bollerslev2002}, \cite{Bollerslev2011} and \cite{Garcia2011}
where the integrated variance based estimations of stochastic volatility models are employed.
These approaches are based on the fact that
although the instantaneous variance process $V$ is latent, the integrated variance is observable and hence we are able to estimate the parameters which belong to the variance process using the integrated variance.
In this paper, we have an additional observable high-frequency based quantity, the realized third moment variation.

Consider again the Heston-type stochastic volatility model in Eqs.~\eqref{Eq:Heston1}~and~\eqref{Eq:Heston2}.
Suppose that total $N$ numbers of the realized second and third moments are observed
and each finite sum of the moment is computed over $[t_i, t_{i+1}]$, $i=0, \ldots, M-1$.
Each $t_i$ is equally distributed with $\Delta = t_{i+1} - t_{i}$.
For the empirical study in the later, we set $\Delta = 1$ day = 1/252 year.

The following moment conditions are considered.
First, the expected value of the quadratic variation of return
\begin{equation}
E[ [R]_{t} ] = \E\left[ \int_0^t V_s \D s \right] = \frac{1-\e^{-\kappa t}}{\kappa} (V_0 -\theta) + \theta t. \label{Eq:mean_rsm}
\end{equation}
Second, the expected value of 1.5 times of the third moment variation
\begin{align}
\frac{3}{2}\E[[R^2,R]_t] = \frac{3\gamma \rho}{\kappa} \left[ V_0\left(\frac{1-\e^{-\kappa t}(\kappa t +1)}{\kappa}\right) + \theta\left(\frac{\e^{-\kappa t}(\kappa t +2)-2}{\kappa} + t \right) \right].\label{Eq:mean_rtm}
\end{align}
Third, the expected value of the square of the integrated variance in Proposition~\ref{Prop:variance}.

With derived moment conditions in Eqs.~\eqref{Eq:mean_rsm},~\eqref{Eq:mean_rtm}~and~\eqref{Eq:mean_sqauredrsm2}, we are able to perform a generalized method of moments (GMM) estimation.
However, before proceeding the GMM, we show a simple estimation method of the stochastic volatility model based on the realized second and third moments.
The method might be useful to find a starting point for the further GMM.
In addition, we will show that the simple estimator itself is good or even better than a complicated GMM.

For the estimation, it is worthwhile to understand the characteristic of each parameter in the stochastic volatility model.
First, the parameter $\theta$ is a long run average of, or unconditional variance since 
$\E [V_s] \rightarrow \theta$ as $s \rightarrow \infty$.
The parameter $\theta$ is estimated by the sample mean of normalized realized second moment
$$\hat \theta_0 = \frac{1}{M} \sum_{i=0}^{M-1} \frac{\widehat{[R]}_{t_i,t_{i+1}}}{\Delta}$$
where $\widehat{[R]}_{t_i,t_{i+1}}$ denotes the realized variance observed over $[t_i, t_{i+1}]$.

Second, the parameter $\kappa$ explains the relationship between the spot variance or integrated variance and its lagged value.
The estimation of $\kappa$ is more difficult than other parameters which can be estimated by the sample means of observed certain quantities such as the realized second or third moments.
The variance process $V$ follows a mean reverting process and its discretized version follows an autoregressive moving average (ARMA) scheme and so does the integrated variance.
Indeed, under the Heston model, by Eqs.~\eqref{Eq:mean_V}~and~\eqref{Eq:mean_rsm},
\begin{equation}
\E_{t_{i-1}} [[R]_{t_i,t_{i+1}}] = \e^{-\kappa \Delta} \E_{t_{i-1}}[[R]_{t_{i-1},t_{i}}] + (1-\e^{-\kappa \Delta})\theta \Delta.\label{Eq:relation_rsm}
\end{equation}

More precisely, note that
\begin{align*}
\D(\e^{\kappa u }V_u) &= \e^{\kappa u} \D V_u + \kappa \e^{\kappa u} V_u \D u \\
&= \kappa \theta \e^{\kappa u} \D u + \gamma \e^{\kappa u} \sqrt{V_u} \D W^v_u
\end{align*}
and by integrating both sides from $t_i$ to $s$,
\begin{align*}
V_{s} = \e^{-\kappa(s-t_i)} (V_{t_i} - \theta) + \theta + \int_{t_i}^{s} \gamma \e^{-\kappa (s -u)} \sqrt{V_u} \D W_u^v.
\end{align*}
If $s=t_{i+1}$, then the above equation is an AR representation for the discretized series of $V$ and we have
\begin{equation}
 V_{t_{i+1}} - \theta = (V_{t_{i}} - \theta)\e^{-\kappa \Delta} + \int_{t_{i}}^{t_{i+1}} \gamma \e^{-\kappa (t_{i+1} - u)}\sqrt{V_u} \D W^v_u.\label{Eq:V_ARMA}
\end{equation}
By integrating $V_s$ from $t_i$ to $t_{i+1}$,
\begin{align}
\int_{t_i}^{t_{i+1}} V_s \D s ={}& [R]_{t_i,t_{i+1}} = \frac{1-\e^{-\kappa \Delta}}{\kappa} (V_{t_i} - \theta) + \theta \Delta + \int_{t_i}^{t_{i+1}}\int_{t_i}^{s} \gamma \e^{-\kappa(s-u)}\sqrt{V_u} \D W_u^v \D s \label{Eq:IV_ARMA}\\
={}&\frac{1-\e^{-\kappa \Delta}}{\kappa} (V_{t_{i-1}} - \theta)\e^{-\kappa \Delta} + \theta \Delta \nonumber\\
&+ \frac{1-\e^{-\kappa \Delta}}{\kappa} \int_{t_{i-1}}^{t_i} \gamma \e^{-\kappa (t_i - u)}\sqrt{V_u} \D W^v_u +  \int_{t_i}^{t_{i+1}}\int_{t_i}^{s} \gamma \e^{-\kappa(s-u)}\sqrt{V_u} \D W_u^v \D s \label{Eq:IV_ARMA2}
\end{align}
where we use Eq.~\eqref{Eq:V_ARMA} for the second equality.
Using the lagged version of Eq.~\eqref{Eq:IV_ARMA}, we have
$$  \frac{1-\e^{-\kappa \Delta}}{\kappa} (V_{t_{i-1}} - \theta) = [R]_{t_{i-1},t_{i}} - \theta \Delta - \int_{t_{i-1}}^{t_{i}}\int_{t_{i-1}}^{s} \gamma \e^{-\kappa(s-u)}\sqrt{V_u} \D W_u^v \D s,$$
and substituting the above equation to Eq.~\eqref{Eq:IV_ARMA2},
\begin{align*}
[R]_{t_i,t_{i+1}} ={}& \e^{-\kappa \Delta} [R]_{t_{i-1},t_{i}} + (1-\e^{-\kappa \Delta})\theta \Delta  +  \frac{1-\e^{-\kappa \Delta}}{\kappa} \int_{t_{i-1}}^{t_i} \gamma \e^{-\kappa (t_i - u)}\sqrt{V_u} \D W^v_u \\ 
&+  \int_{t_i}^{t_{i+1}}\int_{t_i}^{s} \gamma \e^{-\kappa(s-u)}\sqrt{V_u} \D W_u^v \D s - \e^{-\kappa \Delta} \int_{t_{i-1}}^{t_{i}}\int_{t_{i-1}}^{s} \gamma \e^{-\kappa(s-u)}\sqrt{V_u} \D W_u^v \D s\\
={}& \e^{-\kappa \Delta} [R]_{t_{i-1},t_{i}} + (1-\e^{-\kappa \Delta})\theta \Delta \\
&+  \int_{t_i}^{t_{i+1}}\int_{t_i}^{s} \gamma \e^{-\kappa(s-u)}\sqrt{V_u} \D W_u^v \D s 
+ \e^{-\kappa \Delta} \int_{t_{i-1}}^{t_{i}}\int_{s}^{t_i} \gamma \e^{-\kappa(s-u)}\sqrt{V_u} \D W_u^v \D s
\end{align*}
The error term represented by the last two integrations of the above formula is approximately an MA(1) process.
For the detailed information about the ARMA representation of the integrated and realized variance, consult \cite{Meddahi}.

We use this fact to estimate the slope parameter $\e^{-\kappa \Delta}$ of
$$ [R]_{t_i,t_{i+1}} - \overline{[R]} = \e^{-\kappa \Delta} ([R]_{t_{i-1}, t_{i}} - \overline{[R]}) +  n_{i} + b_1n_{i-1}$$
where $\overline{[R]}$ denotes the sample average of the integrated variance and we use $n$ to represent MA(1) error term.
In these days, the estimation procedure of ARMA models is provided by several statistical packages.
Let $\hat \kappa_0$ denote the estimator for $\kappa$ obtained by this method.

Third, the parameter $\gamma$ is for the volatility of the variance.
As $t \rightarrow \infty$, $ \E [V_t^2] -\E [V_t]^2 \rightarrow \frac{\gamma^2 \theta}{2 \kappa}$ by Eq.~\eqref{Eq:mean_squaredV}.
Using this fact and Eq.~\eqref{Eq:mean_sqauredrsm}, we have
$$\E \left[ [R]_t^2 \right] \rightarrow \theta^2 t^2 + \frac{\gamma^2\theta}{\kappa^2}\left(\frac{\e^{-\kappa t}}{\kappa}+t-\frac{1}{\kappa} \right).$$
In other words, this is the expectation of the squared second moment by assuming $V$ starts in negative infinity (or at least quite long time ago) and is at a stationary state at $t$.
Thus, the estimator of $\gamma$ is
$$
\hat\gamma_0 = \sqrt{ \frac{ \frac{1}{M}\sum_{i=0}^{M-1} \widehat{[R]^2}_{t_i,t_{i+1}} - \hat\theta_0^2\Delta^2}{\frac{\hat\theta_0}{\hat\kappa_0^2}\left(\frac{\e^{-\hat\kappa_0 \Delta}}{\hat\kappa_0}+\Delta-\frac{1}{\hat\kappa_0} \right)} }.
$$

Fourth, to estimate the leverage parameter $\rho$, which takes into account the skewness of return, we use the formula of the expectation of the third moment variation derived in the previous section.
Since the long run expectation of the third moment variation is represented by
\begin{equation}
\E\left[[R,R^2]_t\right] = 2\int_0^t \E[R_u V_u] \D u \rightarrow  \frac{2\gamma\rho\theta}{\kappa} \left( \frac{\e^{-\kappa t}}{\kappa} + t - \frac{1}{\kappa} \right)\label{Eq:simple_rtm}
\end{equation}
and hence the leverage parameter $\rho$ is estimated by
$$ \hat \rho_0 = \frac{\frac{1}{M}\sum_{i=0}^{M-1} \widehat{[R,R^2]}_{t_i,t_{i+1}}}{\frac{2\hat\gamma_0\hat\theta_0}{\hat\kappa_0} \left( \frac{\e^{-\hat\kappa_0 \Delta}-1}{\hat\kappa_0} + \Delta \right) }.$$

Now we start a GMM procedure, known to be a consistent and asymptotically normal estimator, to estimate the Heston model based on the exact relationship between the moments and their lagged values.
This approach is little bit different from the previous simple estimation based on long run approximation.
For simplicity, we rewrite
\begin{align}
\E[V_{t_{i+1}} | \F_{t_i}] &= \e^{-\kappa \Delta} V_{t_i}  + \theta (1-\e^{-\kappa \Delta}) =: a V_{t_{i}} + b.\label{Eq:relation_V}
\end{align}
Considering Eq.~\eqref{Eq:mean_rtm}, we put
\begin{align*}
\alpha_3 &= \frac{2\gamma\rho}{\kappa^2}\left( 1- \e^{-\kappa \Delta}(\kappa \Delta + 1) \right), \\
\beta_3 &= \frac{2\gamma\rho}{\kappa} \left( \frac{ \e^{-\kappa \Delta}(\kappa \Delta +2)-2}{\kappa} +\Delta\right) \theta,
\end{align*}
then, by combining the results in Eqs.~\eqref{Eq:mean_rtm}~and~\eqref{Eq:relation_V},
\begin{align}
\E[[R,R^2]_{t_i, t_{i+1}} | \F_{t_i}] = a[R,R^2]_{t_{i-1}, t_{i}} - a \beta_3 + \alpha_3 b + \beta_3\label{Eq:relation_rtm}
\end{align}
Similarly, by considering Eq.~\eqref{Eq:mean_sqauredrsm2}, put
\begin{align*}
C ={}& \frac{1}{\kappa^2}(\e^{-\kappa \Delta} - 1)^2 \\
D ={}& \left\{ -\frac{2}{\kappa}\left(\theta + \frac{\gamma^2}{\kappa} \right)\Delta\e^{-\kappa \Delta} + \frac{4\theta}{\kappa^2}\e^{-\kappa \Delta} - \left(\frac{2\theta}{\kappa^2} + \frac{\gamma^2}{\kappa^3} \right)\e^{-2\kappa \Delta}+ \frac{2\theta}{\kappa}\Delta - \frac{2\theta}{\kappa^2} + \frac{\gamma^2}{\kappa^3}\right\} \\
E ={}& \frac{2\theta}{\kappa}\left(\theta + \frac{\gamma^2}{\kappa} \right)\Delta\e^{-\kappa \Delta} + 2\left(\frac{\gamma^2\theta}{\kappa^3} - \frac{\theta^2}{\kappa^2} \right)\e^{-\kappa \Delta} + \left(\frac{\theta^2}{\kappa^2} + \frac{\gamma^2\theta}{2\kappa^3} \right)\e^{-2\kappa \Delta}\\
&+ \theta^2 \Delta^2 + \left(\frac{\gamma^2 \theta}{\kappa^2} - \frac{2\theta^2}{\kappa} \right)\Delta + \frac{\theta^2}{\kappa^2} - \frac{5\gamma^2\theta}{2\kappa^3}
\end{align*}
and with Eq.~\eqref{Eq:mean_squaredV}, put
\begin{align*}
c &= \e^{-2\kappa \Delta}\\
d &= \frac{2\kappa\theta + \gamma^2}{\kappa}(\e^{-\kappa \Delta} - \e^{-2\kappa \Delta}) \\
f &= \frac{\theta(2\kappa\theta + \gamma^2)}{\kappa}\left( \frac{1}{2} - \e^{-\kappa \Delta} + \frac{1}{2}\e^{-2\kappa \Delta}\right).
\end{align*}
Then we have
\begin{align}
\E\left[ [R]^2_{t_{i+1},t_{i+2}} | \F_{t_{i+1}} \right] ={}& c [R]^2_{t_i,t_{i+1}} + \frac{Cd + (a-c)D}{\alpha} [R]_{t_i,t_{i+1}}\nonumber\\
&-\frac{\beta}{\alpha} (Cd + (a-c)D) + Cf + bD + (1-c)E \nonumber\\
=:{}& c [R]^2_{t_i,t_{i+1}} + F [R]_{t_i,t_{i+1}} + G \label{Eq:relation_sqrsm}
\end{align}
where $\alpha = (1-\e^{-\kappa \Delta})/\kappa$ and $\beta=(1-\alpha)\theta\Delta$.

Now we have three exact relationships in Eqs.~\eqref{Eq:relation_rsm},~\eqref{Eq:relation_rtm}~and~\eqref{Eq:relation_sqrsm} between the variations and their lagged values.
Note that these depends on the successive relationships of the variations contrast with the simple method where the method largely depends on the sample averages of the variations.
With these relationships, we consider an objective function of $\eta = \{\kappa, \theta, \gamma, \rho\}$ to perform a generalized method of moments estimation:
\[
g(\eta) = \left[
  \begin{array}{c}
    [R]_{t_{i+1},t_{i+2}} - \e^{-\kappa \Delta} [R]_{t_{i},t_{i+1}} - (1-\e^{-\kappa \Delta})\theta\Delta\\
    \{ [R]_{t_{i+1},t_{i+2}} - \e^{-\kappa \Delta} [R]_{t_{i},t_{i+1}} - (1-\e^{-\kappa \Delta})\theta\Delta\} [R]_{t_{i-1},t_{i}}\\
    \{ [R]_{t_{i+1},t_{i+2}} - \e^{-\kappa \Delta} [R]_{t_{i},t_{i+1}} - (1-\e^{-\kappa \Delta})\theta\Delta\} [R]_{t_{i-2},t_{i-1}}\\
    \left. [R,R^2]_{t_{i+1}, t_{i+2}} - a[R,R^2]_{t_{i}, t_{i+1}} + a \beta_3 - \alpha_3 b - \beta_3 \right. \\
    ([R,R^2]_{t_{i+1}, t_{i+2}} - a[R,R^2]_{t_{i}, t_{i+1}} + a \beta_3 - \alpha_3 b - \beta_3)[R,R^2]_{t_{i-1}, t_{i}} \\
    \left. [R]^2_{t_{i+1},t_{i+2}} - c [R]^2_{t_i,t_{i+1}} - F [R]_{t_i,t_{i+1}} - G \right.
  \end{array}
\right]
\]
The lagged values of $[R]$ and $[R,R^2]$ are used as instrumental variables as in the GMM estimation of ARMA(1,1) models.

Table~\ref{Table:Estimation} presents the estimation results of Heston's stochastic volatility model and the S\&P 500 return series from 2003 to 2007.
In the second and fifth columns of the table, we have presumed parameter values of the models, and the third and sixth columns report the corresponding estimates by the simple method,
and the fourth and seventh columns report the estimates by the GMM.
The last two columns present the estimates of parameters by the simple and GMM methods under the assumption that the S\&P 500 return series follows the stochastic volatility model.
The parenthesis are for standard errors except for $\gamma$ where the the standard errors for $\gamma^2$ are reported.
For the GMM with the S\&P 500 data, we use the results of the simple estimation as a starting point.

\begin{table}
\centering
\caption{Estimation results of Heston's model and the S\&P 500 return}\label{Table:Estimation}
\begin{tabular}{ccccccc|cc}
  \hline
  & \multicolumn{5}{c}{Heston's model} & &\multicolumn{2}{c}{S\&P 500}\\
  & Model I & Simple & GMM & Model II & Simple & GMM & Simple & GMM\\
  \hline
  $\theta$ & 0.05 & 0.0493 & 0.0475 & 0.02 & 0.0204 & 0.0202 & 0.0196 & 0.0194 \\
  & & & (0.0076) & & & (0.0022) & & (0.0056)\\
  $\kappa$ & 5 & 6.1466 & 7.2500 & 15 & 14.320 &  13.404 & 10.047 & 9.28\\
  & & & (1.7295) & & & (2.6344) & & (13.909)\\
  $\gamma$ & 0.8 & 0.7708 & 0.5792 & 0.7 & 0.6571 & 0.5075 & 0.8483 & 0.4168\\
  & & & (0.0065) & & & (0.0422) & & (0.2216)\\
  $\rho$ & -0.5 & -0.6719 & -0.7937 & 0.3 & 0.3957 &  0.5551 & -0.6189 & -0.9891\\
  & & & (4.9118) & & & 1.5832) & & (5.6924)\\
  \hline
\end{tabular}
\end{table} 

The results of the simple methods are quite close to the original parameter settings in the simulation study.
The estimation results of the GMM are not much improved compared with the result of the simple method
and the GMM estimates are even worse for the parameters $\gamma$ and $\rho$.
This is because the parameters $\theta$, $\gamma$ and $\rho$ take into account the long run property the processes.
Taking sample average in the simple method have a better performance compared with taking successive relation in GMM method,
since the successive relations are not informative enough to explain the long run expectations.
Thus, it is hard to improve the estimation performance when the estimation is based on successive relationships such as Eqs.~\eqref{Eq:relation_rtm} and \eqref{Eq:relation_sqrsm}.
Furthermore, the estimation results of the GMM is sensitive to the choice of the instrumental variables.
We conclude that the simple estimation method based on long run expectations of the variations are good enough for the estimation of parameters $\theta$, $\gamma$ and $\rho$.
For the estimation of $\kappa$, we need a successive relation of the second moment and the conventional ARMA estimation is enough for the estimation.

The estimate of $\rho$ of the S\&P 500 by our simple method is larger in magnitude than the estimate of $\rho$ in \cite{Bollerslev2002}, which is $-0.0243$.
However, the data in the paper is the Deutsch Mark U.S. Dollar spot exchange rate from 1986 to 1996 which is different from ours.
Our result is rather similar to the result of \cite{Ait2013} where they used the S\&P 500 data ranged from 2004 to 2007 and the estimate of $\rho$ is $-0.77$.

In \cite{Bollerslev2002} where the estimate of $\kappa$ is around 0.14 but the unit of the model in the paper is on a daily basis but our model is on a yearly basis.
It is useful to compare $\exp(-\kappa \Delta)$ which explains the relation between 
$[R]_{t_i,t_{i+1}}$ and $[R]_{t_{i-1},t_{i}}$ as in Eq.~\eqref{Eq:relation_rsm}.
By matching the scale, in our simple method, $-\kappa\Delta =0.0399$ and $\exp(-\kappa \Delta) = 0.9609$ and in \cite{Bollerslev2002}, $\exp(-\kappa \Delta) = 0.8637$.
This implies that there is a weaker mean reverting property (i.e., stronger persistence) in the variance process in our empirical analysis.

\section{Conclusion and future work}\label{Sect:conclusion}
The basic properties of the third and fourth moment variations are derived.
We proposed two methods of estimation for a stochastic volatility model based on the realized second and third moments.
One is a simple method based on sample average of long run moments and ARMA type estimation.
The other is based on the GMM and exact relationship among second and third moments.
The simple method of estimation shows good performance compared with the complicated GMM.
Our strategy can be applied to other affine models and an interesting extension would be a multidimensional affine model in finance and economics with multiasset cases.

\bibliography{realized_moments}
\bibliographystyle{apalike}

\appendix
\section{Moment condition for the fourth moment variation}\label{Sect:fmv}

We consider the moment condition of $\E[[R^2]_t]$.
Note that
\begin{align*}
\E[R_u^2 V_u] &= \E \left[ V_u \left(2\int_0^u R_s\D R_s + \int_0^u \D [R]_s\right) \right] \\
&= \E \left[ \left( V_0 + \int_0^u \kappa (\theta -V_s) \D s + \int_0^u \sqrt{V_s} \D W^v_s \right ) \int_0^u R_s\sqrt{V_s}\D W^s_s \right] + \E\left[V_u \int_0^u V_s \D s \right]\\
&= -\kappa \E\left[\int_0^u V_s \D s \int_0^u R_s\sqrt{V_s}\D W^s_s \right] + \E\left[\int_0^u R_s V_s \D s\right] + \E\left[V_u \int_0^u V_s \D s \right]. 
\end{align*}
For the second and third terms, we already derive in Propositions~\ref{Prop:variance}~and~\ref{Prop:tmv}.
For the first term of the r.h.s., we have
\begin{align*}
m(u) := \E\left[\int_0^u V_s \D s \int_0^u R_s\sqrt{V_s}\D W^s_s \right] = \E\left[ \int_0^u V_s \left( \int_0^s R_k \sqrt{V_k} \D W^s_k\right) \D s \right]
\end{align*}
and
\begin{align*}
&\E \left[ V_s \int_0^s  R_k \sqrt{V_k} \D W^s_k \right]\\
&= \E \left[ \left(V_0 + \int_0^s \D V_k \right) \int_0^s  R_k \sqrt{V_k} \D W^s_k \right]\\
&= \E\left[\int_0^s \kappa (\theta - V_k) \D k \int_0^s R_k \sqrt{V_k} \D W^s_k\right] + \E
\left[\int_0^s \gamma \sqrt{V_k} \D W^v_k \int_0^s R_k \sqrt{V_k} \D W^s_k \right]\\
&= -\kappa \E\left[\int_0^s V_k \D k \int_0^s R_k \sqrt{V_k}\D W^s_k \right] + \E
\left[\gamma\rho \int_0^s R_k V_k\D k \right]\\
&= -\kappa m(s) + \frac{\gamma^2 \rho^2}{\kappa} \left\{(V_0-\theta)\left( \frac{1-\e^{-\kappa s}(\kappa s +1)}{\kappa} \right) + \theta \left(\frac{\e^{-\kappa s}-1}{\kappa} + s \right) \right\}
\end{align*}
and hence
$$ m(u) = \int_0^u \left[ -\kappa m(s) + \frac{\gamma^2 \rho^2}{\kappa} \left\{(V_0-\theta)\left( \frac{1-\e^{-\kappa s}(\kappa s +1)}{\kappa} \right) + \theta \left(\frac{\e^{-\kappa s}-1}{\kappa} + s \right) \right\} \right] \D s$$
and solving the integration equation, we have
\begin{align*}
m(u) = \frac{\gamma^2 \rho^2}{\kappa^3}\left[ (V_0-\theta)\left\{ 1-\e^{-\kappa u} - \left(\kappa u + \frac{\kappa^2}{2}u^2 \right)\e^{-\kappa u}\right\} + \theta\left\{ 2(\e^{-\kappa u} - 1) +  \kappa u (1 + \e^{-\kappa u})\right\}\right].
\end{align*}
Thus, by using Propositions~\ref{Prop:variance}~and~\ref{Prop:tmv}, we derive 
\begin{align*}
\E[R_u^2 V_u]={}& \frac{\gamma^2 \rho^2}{\kappa^3}\left[ (V_0-\theta)\left\{ 1-\e^{-\kappa u} - \left(\kappa u + \frac{\kappa^2}{2}u^2 \right)\e^{-\kappa u}\right\} + \theta\left\{ 2(\e^{-\kappa u} - 1) +  \kappa u (1 + \e^{-\kappa u})\right\}\right]\\
&+ (V_0-\theta)\left(\theta + \frac{\gamma^2}{\kappa}\right)u\e^{-\kappa u} + \left(\frac{(V_0-\theta)(V_0-2\theta)}{\kappa} - \frac{\gamma^2 V_0}{\kappa^2}\right)\e^{-\kappa u} \\
&+ \left\{ -\frac{(V_0-\theta)^2}{\kappa} + \frac{\gamma^2}{\kappa^2}\left(V_0 -\frac{1}{2}\theta \right) \right\}\e^{-2\kappa u} + \theta^2 u + \frac{(V_0-\theta)\theta}{\kappa} + \frac{\gamma^2\theta}{2\kappa^2}\\
&+\frac{\gamma \rho}{\kappa} \left[(V_0-\theta)\left( \frac{1-\e^{-\kappa u}(\kappa u +1)}{\kappa} \right) + \theta \left(\frac{\e^{-\kappa t}-1}{\kappa} + u \right) \right].
\end{align*}
In addition, by integrating $\E[R_u^2 V_u]$ with respect to $u$, the moment condition for the fourth moment variation is derived.

\section{The variance of the third moment variation}\label{Sect:vtmv}
Assume that the stochastic volatility model of Eqs.~\eqref{Eq:Heston1}~and~\eqref{Eq:Heston2}.
The drift term in the volatility usually leads to very complicated computation to derive the variance of the third moment.
The third moment variation over $[0,t]$ is
$$ [R^2, R]_t = 2\int_0^t R_u V_u \D u$$
and the variance of the third moment is
\begin{align*}
\mathrm{Var}([R^2, R]_t) = 4\E\left[\left(\int_0^t R_u V_u \D u \right)^2 \right] - 4\left(\E\left[\int_0^t R_u V_u \D u  \right]\right)^2.
\end{align*}
The second term in the r.h.s. is derived in Proposition~\ref{Prop:tmv}.
Let $s\vee u$ denote the minimum of $s$ and $u$.
To compute the first term of the r.h.s., we have
\begin{align*}
&\E\left[\left(\int_0^t R_u V_u \D u \right)^2 \right] = \E\left[\int_0^t \int_0^t R_u V_u R_s V_s \D u \D s \right] =\int_0^t \int_0^t \E [R_u V_u R_s V_s] \D u \D s \\
&=\int_0^t \int_0^t \E [R_{u\vee s} V_{u\vee s}(R_{u\vee s}+\Delta R)( V_{u\vee s} + \Delta V)]\D u \D s\\
&=\int_0^t \int_0^t \E [R_{u\vee s}^2 V_{u\vee s}^2 + R_{u\vee s}^2 V_{u\vee s}\Delta V + R_{u\vee s}V_{u\vee s}^2 \Delta R + R_{u\vee s} V_{u\vee s} \Delta R  \Delta V ]\D u \D s.
\end{align*}
For the expectation of each term in the integrand, we have
\begin{align*}
\E [R_{u\vee s}^2 V_{u\vee s}\Delta V] &= \E [R_{u\vee s}^2 V_{u\vee s} \E[ \Delta V | \F_{s\vee u}]] = \E[R_{u\vee s}^2 V_{u\vee s}(V_{u\vee s}-\theta)(\e^{-\kappa|s-u|}-1)] \\
&= \E[[R_{u\vee s}^2 V_{u\vee s}^2](\e^{-\kappa|s-u|}-1) - \E[[R_{u\vee s}^2 V_{u\vee s}]\theta (\e^{-\kappa|s-u|}-1)
\end{align*}
where we use Eq.~\eqref{Eq:relation_V}, and
$$ \E [R_{u\vee s}V_{u\vee s}^2 \Delta R] = 0$$
and
\begin{align*}
\E[R_{u\vee s} V_{u\vee s} \Delta R  \Delta V] ={}& \E[\E[R_{u\vee s} V_{u\vee s} \Delta R  \Delta V | \F_{s\vee u}] ]\\
={}& \rho \gamma  \E[ R_{u\vee s} V_{u\vee s} \{ (V_{u\vee s} - \theta)|u-s|\e^{-\kappa|u-s|} + \theta(V_{u\vee s} - \theta + 1)|u-s| \} ]\\
={}& \E[ R_{u\vee s} V_{u\vee s}^2] \rho \gamma |u-s|(\theta+\e^{-\kappa|u-s|}) \\
&+ \E[ R_{u\vee s} V_{u\vee s}]\rho \gamma \theta |u-s|(1-\theta-\e^{-\kappa|u-s|})
\end{align*}
where we use Eq.~\eqref{Eq:RV}.
Therefore,
\begin{align*}
&\E\left[\left(\int_0^t R_u V_u \D u \right)^2 \right] \\
&=\int_0^t \int_0^t \left\{ \E[[R_{u\vee s}^2 V_{u\vee s}^2]\e^{-\kappa|s-u|} - \E[[R_{u\vee s}^2 V_{u\vee s}]\theta (\e^{-\kappa|s-u|}-1) \right.\\
&\quad \left. + \E[ R_{u\vee s} V_{u\vee s}^2] \rho \gamma |u-s|(\theta+\e^{-\kappa|u-s|}) 
+ \E[ R_{u\vee s} V_{u\vee s}]\rho \gamma \theta |u-s|(1-\theta-\e^{-\kappa|u-s|}) \right\} \D u \D s.
\end{align*}
The above representation shows that the variance of the third moment variation depends on the expected values of $E[[R_{u\vee s}^2 V_{u\vee s}^2]]$, $E[[R_{u\vee s}^2 V_{u\vee s}]]$, $E[[R_{u\vee s} V_{u\vee s}^2]]$ and $E[[R_{u\vee s} V_{u\vee s}]]$.

\end{document}